\documentclass[letterpaper, 10pt, conference]{ieeeconf} 
\IEEEoverridecommandlockouts
\usepackage{graphicx}
 \usepackage{resizegather}
\usepackage{url}
\usepackage{xcolor}
\usepackage{amsmath}
\usepackage{mathtools}
 \mathtoolsset{showonlyrefs}
\usepackage{tikz}
\usepackage{verbatim}
\usepackage{subcaption}
\usepackage{pgfplots}
\pgfplotsset{compat=1.10}
\usepackage{graphicx}
\usepackage{epstopdf}
\usepackage{amssymb}
\usepackage{relsize}
\usepackage{algorithm}
\definecolor{Gray}{gray}{0.90}
\usepackage{colortbl}

\usepackage{algorithm}
\usepackage{algorithmicx}
\usepackage{algcompatible}
\usepackage{algpseudocode}
\usepackage{color}
\usepackage{blkarray}

\newtheorem{theorem}{\bf{Theorem}}[section]

\newtheorem{remark}[theorem]{\bf{Remark}}
\newenvironment{definition}[1][Definition]{\begin{trivlist}
\item[\hskip \labelsep {\bfseries #1}]}{\end{trivlist}}

\usepackage{amssymb}
\usepackage{url}

\newcommand{\calN}[0]{\mathcal{N}}

\newcommand{\calD}[0]{\mathcal{D}}

\definecolor{amber}{rgb}{1.0, 0.75, 0.0}

\begin{document}

\title{\bf Strong Structural Controllability of Diffusively Coupled Networks: Comparison of Bounds Based on Distances and Zero Forcing}
\author{ Yasin Yaz{\i}c{\i}o\u{g}lu, Mudassir Shabbir, Waseem Abbas, and Xenofon Koutsoukos
\thanks{Yasin~Yaz{\i}c{\i}o\u{g}lu is with the Department of Electrical and Computer Engineering at the University of Minnesota, Minneapolis, MN, USA. Email: ayasin@umn.edu} 
\thanks{Mudassir~Shabbir is with the Computer Science Department at the Information Technology University, Lahore, Punjab, Pakistan. Email: mudassir@rutgers.edu}
\thanks{Waseem~Abbas and Xenofon Koutsoukos are with the Department of Electrical Engineering and Computer Science, Vanderbilt University, Nashville, TN, USA (e-mails: waseem.abbas@vanderbilt.edu, xenofon.koutsoukos@vanderbilt.edu).}
}

\maketitle
\begin{abstract}
We study the strong structural controllability (SSC) of
diffusively coupled networks, where the external control
inputs are injected to only some nodes, namely the leaders. For such systems, one measure of controllability is the dimension of strong structurally controllable subspace, which is equal to the smallest possible rank of controllability matrix under admissible (positive) coupling weights. In this paper, we compare two tight lower bounds on the dimension of strong structurally controllable subspace: one based on the distances of followers to leaders, and the other based on the graph coloring process known as zero forcing. We show that the distance-based lower bound is usually better than the zero-forcing-based bound when the leaders do not constitute a zero-forcing set. On the other hand, we also show that any set of leaders that can be shown to achieve complete SSC via the distance-based bound is necessarily a zero-forcing set. These results indicate that while the zero-forcing based approach may be preferable when the focus is only on verifying complete SSC, the distance-based approach is usually more informative when partial SSC is also of interest. Furthermore, we also present a novel bound based on the combination of these two approaches, which is always at least as good as, and in some cases strictly greater than, the maximum of the two bounds. We support our analysis with numerical results for various graphs and leader sets.

\end{abstract}


\section{Introduction}

Networks of diffusively coupled agents, where each node's state is attracted toward the weighted average of its neighbors' states, appear in numerous systems such as sensor networks, distributed robotics, power grids, social networks, and biological systems. Such systems are often modeled by using their interaction graphs where the nodes represent the agents, and the weighted edges denote the couplings among agents. One major research question regarding such systems is whether a desired global behavior can be induced by injecting external inputs to only a subset of agents, so called the leaders. This question has motivated numerous studies on relating \textit{network controllability} to the structure of the interaction graph. Various graph theoretic tools have been utilized to provide topology-based characterizations of network controllability. Examples include \textit{equitable partitions} (e.g., \cite{Rahmani09}), \textit{maximum matchings} (e.g., \cite{Liu11,Chapman13}), \textit{centrality based measures} (e.g., \cite{Liu12,Pan14}), \textit{dominating sets} (e.g., \cite{Nacher14}), \textit{distances} (e.g.,  \cite{Zhang14,Yasin12CDC,yazicioglu2016graph,Van2017distance}), and \textit{zero forcing} (e.g., \cite{work2008zero,monshizadeh2014zero,trefois2015zero,monshizadeh2015strong,mousavi2018structural}). 



In this paper, we focus on the strong structural controllability of diffusively coupled networks. More specifically, we consider the dimension of \emph{strong structurally controllable subspace (SSCS)}, i.e., the minimum possible rank of controllability matrix under weighted Laplacian dynamics, as a measure of controllability. Two graph theoretic concepts are  known to yield a tight lower bound on this measure: distances and zero forcing. In this paper, we first compare these two approaches. We characterize various cases where the distance-based lower bound is greater than the zero-forcing-based bound. On the other hand, we also show that, for any network of $n$ nodes, any set of leaders that makes the distance-based bound equal to $n$ is necessarily a zero forcing set, i.e., they also make the zero-forcing-based bound equal to $n$. These results indicate that while the zero-forcing-based approach is better for verifying complete strong structural controllability, the distance-based approach is usually more informative when the leaders do not constitute a zero forcing set. We also propose a novel bound based on the combination of these two methods, which is always at least as good as, and in some cases greater than, the maximum of the two bounds. Finally, we support our analysis with some numerical results. 


The organization of this paper is as follows: Section \ref{prelim} provides some preliminaries. Section \ref{compare} presents our results regarding the comparison of bounds. Section \ref{combine} provides a novel bound based on the combination of distance-based and zero-forcing-based methods. Some numerical results are given in Section \ref{sims}. Finally, Section \ref{conc} concludes the paper. 



\section{Preliminaries}
\label{prelim}
\subsection{Graph Basics}

We consider a network represented by a simple directed graph $G=(V,E)$ where the node set $V = \{v_1,v_2,\ldots,v_n \}$ represent agents, and the edge set $E$ represents interconnections between agents. An edge from a node $v_i\in V$ to a node $v_j\in V$ is denoted by $e_{ij}$. The \emph{out-neighborhood} of node $v_i$ is ${\calN_i\triangleq\{v_j\in V: e_{ij} \in E\}}$. The \emph{in-neighborhood} of node $v_i$ is ${\calN_i\triangleq\{v_j\in V: e_{ji} \in E\}}$. The \emph{distance} $d(v_i,v_j)$, is simply the number of edges on the shortest path from $v_i$ to $v_j$. Accordingly, $d(v_i,v_i)=0$ and $d(v_i,v_j)=\infty$ if there is no path from $v_i$ to $v_j$. The graph is strongly connected if there is a path from any node to any other node. The \emph{weight function} $w:E\rightarrow \mathbb{R}^+$
assigns a positive weight $w(e_{ij})$ to each edge $e_{ij}$, which will denote how strongly $v_i$ is influenced by $v_j$ in the dynamical model below.


\subsection{System Model}

While the model can easily be extended to agents with higher-dimensional states, for the sake of simplicity let each agent ${v_i\in V}$ have a state $x_i\in\mathbb{R}$. The overall state of the system is ${x = \left[\begin{array}{cccc}x_1 & x_2 & \cdots & x_n\end{array}\right]^T\in\mathbb{R}^n}$. 
The states evolve under the weighted Laplacian dynamics, 
\begin{equation}
\label{eq:dynamics}
\dot{x} = -L_wx + B u,
\end{equation}
where $L_w\in\mathbb{R}^{n\times n}$ is the weighted \emph{Laplacian matrix} of $G$ and is defined as $L_w = \Delta - A_w$.  Here, $A_w\in\mathbb{R}^{n\times n}$ is the weighted \emph{adjacency matrix} defined as
\begin{equation}
\label{eq:A_matrix}
\left[ A_w \right]_{ij} =
\left\lbrace
\begin{array}{ccc}
w(e_{ij}) & \text{if } e_{ij}\in E,\\
0 	   & \text{otherwise,}
\end{array}
\right.
\end{equation}
and $\Delta\in\mathbb{R}^{n\times n}$ is the \textit{degree matrix} whose entries are 
\begin{equation}
\label{eq:Delta_matrix}
\left[ \Delta \right]_{ij} =
\left\lbrace
\begin{array}{ccc}
\sum_{k = 1}^{n} A_{ik}& \text{if } i=j\\
0 	   & \text{otherwise.}
\end{array}
\right.
\end{equation}
The matrix $B\in\mathbb{R}^{n\times m}$ in \eqref{eq:dynamics} is an \textit{input matrix}, where $m$ is the number of  leaders (inputs), which are the nodes to which an external control signal is applied. Let ${V_\ell = \{\ell_1,\ell_2,\cdots,\ell_m\}\subseteq V}$
be the set of \emph{leaders}, then
\begin{equation}
\label{eq:B_matrix}
\left[ B \right]_{ij} =
\left\lbrace
\begin{array}{ccc}
1& \text{if } v_i = \ell_j\\
0 	   & \text{otherwise.}
\end{array}
\right.
\end{equation}

\subsection{Strong Structural Controllability}
\label{sec:SSC}
A state $x_f\in\mathbb{R}^n$ is a \emph{reachable state} if there exists an input $u$ that can drive the network in \eqref{eq:dynamics} from the origin to $x_f$ in a finite amount of time. A network $G=(V,E)$ in which edges are assigned weights according to the weight function $w$, and contains $V_\ell\subseteq V$ leaders is called \emph{completely controllable} if every point in $\mathbb{R}^n$ is reachable. Complete controllability can be checked via the rank of \emph{controllability matrix}, i.e.,
\begin{equation}
\small
\Gamma(L_w,V_\ell) =
\left[
\begin{array}{ccccc}
B & (-L_w)B & (-L_w)^2B & \cdots & (-L_w)^{n-1}B
\end{array}
\right],
\end{equation}
where $B$ is defined as in \eqref{eq:B_matrix}. The network is completely controllable if and only if the rank of $\Gamma(L_w,V_\ell)$ is $n$, and in such case $(L_w,B)$ is called a \emph{controllable pair}. Note that edges in $G$ define the \emph{structure}---location of zero and non-zero entries in the Laplacian matrix---of the underlying graph, for instance, see Fig.~\ref{fig:Lap}. For any given graph $G=(V,E)$ and $V_\ell$ leaders, the rank of controllability matrix depends on the weights assigned to edges.

\begin{figure}[htb]
\centering
\includegraphics[scale=0.9]{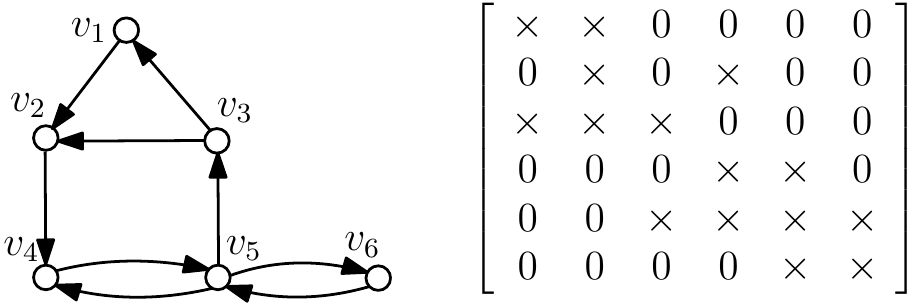}
\caption{A graph and its structured Laplacian whose non-zero off-diagonal entries are positive and rows sum to zero.}
\label{fig:Lap}
\end{figure}

A network $G=(V,E)$ with $V_\ell$ leaders is \emph{strong structurally controllable} if $(L_w,B)$ is a controllable pair for any choice of weight function $w$. The \emph{dimension of strong structurally controllable subspace (SSCS)}, denoted by $\gamma(G,V_\ell)$, is the smallest possible rank of controllability matrix under feasible weights, i.e.,
\begin{equation}
\label{eq:dim_SSC}
\gamma(G,V_\ell) = \min\limits_{w:E\rightarrow \mathbb{R}^+} \left(\text{rank}\;\Gamma(L_w,V_\ell)\right),
\end{equation}
where the minimum is taken over all feasible weight functions $w:E\rightarrow \mathbb{R}^+$. Roughly, $\gamma(G,V_\ell)$ quantifies how much of the network can be controlled through the leaders $V_\ell$ under any feasible choice of edge weights.

\begin{remark}
The original notion of strong structural controllability \cite{mayeda1979strong} considers the worst-case controllability under any allocation of the non-zero values in a system's structure matrix. Our focus is on the worst-case controllability under weighted Laplacian dynamics, which narrows down the feasible set of system matrices. However, it is worth mentioning that both the distance-based and the zero-forcing-based methods, which will be explained next, are actually applicable to more generalized dynamics (e.g., \cite{Van2017distance,monshizadeh2015strong}).
\end{remark}

\subsection{Distance-based Lower Bound: $\delta(G,V_\ell)$}
\label{sec:PMI_bound}
 Given a network with $m$ leaders $V_\ell = \{\ell_1,\cdots,\ell_m\}$, we define the \emph{distance-to-leaders} (DL) vector of each $v_i\in V$ as
\begin{equation*}
\label{eq:DLvector}
D_i = \left[
\begin{array}{ccccc}
d(v_i, \ell_1) & d(v_i,\ell_2)  & \cdots & d(v_i,\ell_m)
\end{array}
\right]^T \in \mathbb{Z}^m.
\end{equation*}
The $j^{th}$ component of $D_i$, denoted by $[D_i]_j$, is equal to the distance of $v_i$ to $\ell_j$. Next, we provide the definition of  \emph{pseudo-monotonically increasing} sequences of DL vectors.

\begin{definition}{(\emph{Pseudo-monotonically Increasing (PMI) Sequence)}}
\label{def:PMI}
A sequence of distance-to-leaders vectors $\calD$ is PMI if for a vector $\calD_i$ in the sequence, there exists some $\pi(i)\in\{1,2,\cdots,m\}$ such that
\begin{equation}
\label{eq:PMIcondition}
[\calD_i]_{\pi(i)} < [\calD_{j}]_{\pi(i)}, \;\;\forall j > i.
\end{equation}
We say that $\calD_i$ satisfies the \emph{PMI property} at coordinate $\pi(i)$ whenever $[\calD_i]_{\pi(i)} < [\calD_{j}]_{\pi(i)},\;\forall j>i$.
\end{definition}

An example of DL vectors is illustrated in Fig. \ref{fig:PMI}, where a PMI sequence of length six can be constructed as

\begin{equation}
\label{eq:PMIexample}
\calD =
\left\{
\left[
\begin{array}{c}
3\\ \textcircled{0}
\end{array}
\right]
,
\left[
\begin{array}{c}
\textcircled{0} \\ 4
\end{array}
\right]
,
\left[
\begin{array}{c}
\textcircled{1} \\ 4
\end{array}
\right]
,
\left[
\begin{array}{c}
2\\ \textcircled{1}
\end{array}
\right]
,
\left[
\begin{array}{c}
\textcircled{3} \\ 2
\end{array}
\right],
\left[
\begin{array}{c}
4 \\ 3
\end{array}
\right]
\right\}.
\end{equation}

Indices of circled values in \eqref{eq:PMIexample} are the coordinates, $\pi(i)$, at which the corresponding distance-to-leaders vectors are satisfying the PMI property.
The longest PMI sequence of distance-to-leaders vectors is related to the dimension of SSCS as stated in the following result.
\begin{figure}[htb]
\centering
\includegraphics[scale=0.9]{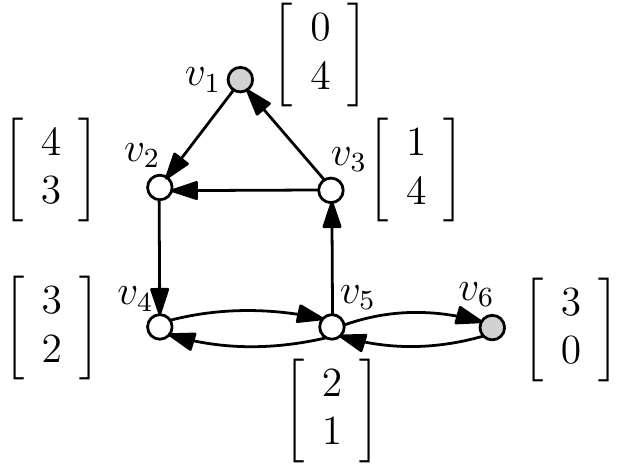}
\caption{A network with two leaders, $V_\ell = \{v_1,v_6\}$, and the corresponding distance-to-leaders (DL) vectors.}
\label{fig:PMI}
\end{figure}
\begin{theorem} \cite{yazicioglu2016graph}
\label{thm:PMI}
Consider any network $G=(V,E)$ with the leaders $V_\ell \subseteq V$. Let $\delta(G,V_\ell)$ be the length of longest PMI sequence of distance-to-leaders vectors with at least one finite entry.  Then,
\begin{equation}
\label{eq:PMI_bound}
\delta(G,V_\ell) \le \gamma(G,V_\ell).
\end{equation}
\end{theorem}

\begin{remark}
While the bound in \eqref{eq:PMI_bound} was presented for connected undirected graphs in Theorem 3.2 in \cite{yazicioglu2016graph}, it also holds for any choice of leaders on strongly connected graphs as shown in Remark 3.1 in \cite{yazicioglu2016graph}. Such connectivity properties already ensure that all DL vectors have only finite entries. The bound can easily be extended to graphs without strong connectivity by excluding the DL vectors of all $\infty$, which belong to followers that can not be influenced by any leader.
\end{remark}




 
 


\subsection{Zero-forcing-based Lower Bound: $\zeta(G,V_\ell)$}
We first give the definitions of \emph{zero forcing process} and \emph{derived set}.

\begin{definition}{\em (Zero Forcing Process)}
Given a graph $G=(V,E)$ where each node is initially colored either \textit{white} or \textit{black}, zero forcing process is defined by the following coloring rule: if $v\in V$ is colored black and has exactly one white in-neighbor $u$, then the color of $u$ is changed to black and $u$ is said to be \emph{infected} by $v$.
\end{definition}

\begin{definition}{\em (Derived Set)}
Given an initial set of black nodes $V'\subseteq V$ (called the \textit{input set}) in a graph $G=(V,E)$, there exists a unique derived set, $\text{dset}(G,V') \subseteq V$, which is the resulting set of black nodes when no further color changes are possible under the zero forcing process. An input set $V'$ is called a \emph{zero forcing set} (ZFS) if $\text{dset}(G,V')=V$. 
\end{definition}

\begin{theorem}\cite{monshizadeh2015strong}
For any network $G=(V,E)$ with the leaders $V_\ell \subseteq V$,
\begin{equation}
\label{eq:z_bound}
\zeta(G,V_\ell) \le \gamma(G,V_\ell),
\end{equation}
where $\zeta(G,V_\ell)= |\text{dset}(G,V_\ell)|$ is the size of the derived set corresponding to the input set $V_\ell$.
\end{theorem}
\begin{proof}
Proof follows from Lemma 4.2 in \cite{monshizadeh2015strong}, which shows that for a set of state matrices including weighted Laplacians as a subset, the controllable subspace always contains a $|\text{dset}(G,V_\ell)|$-dimensional subspace. 
\end{proof}

\subsection{Computation of the Bounds}
For any given network with $n$ nodes and $m$ leaders, all pair-wise distances can be computed in $O(n^3)$ time (e.g., \cite{floyd1962algorithm}). Given the distances, $\delta(G,V_\ell)$ can be computed in ${O(m(n\log n + n^m))}$ time \cite{shabbir2019computation}. When the number of leaders makes this computation intractable, an approximation (underestimation), which was shown to be very close to the exact value on various networks, can be obtained in $O(mn\log n)$ time \cite{shabbir2019computation}. On the other hand, $\zeta(G,V_\ell)$ can be computed in ${O(n^2)}$ time by  recursively applying the coloring rule to the in-neighbors of infected nodes until no further color change is possible.




\section{Comparison of Bounds}
\label{compare}
In this section, we compare the distance-based bound, $\delta(G,V_\ell)$, and the zero-forcing-based bound, $\zeta(G,V_\ell)$. 
It is worth mentioning that both $\delta(G,V_\ell)$ and $\zeta(G,V_\ell)$ are tight bounds. For instance, in the case of undirected graphs, any path graph in which one of the end nodes is a leader, or any cycle graph in which two adjacent nodes are leaders satisfy ${\zeta(G,V_\ell) = \delta(G,V_\ell) = \gamma(G,V_\ell) = n}$. Furthermore, neither of these two tight bounds is guaranteed to be at least as good as the other in all possible cases. We provide one example for $\zeta(G,V_\ell)>\delta(G,V_\ell)$ and one example for $\delta(G,V_\ell)>\zeta(G,V_\ell)$ in Fig.  \ref{fig:pvszex}. Accordingly, we aim to identify when one bound may be preferable to the other.
\begin{figure}[htb]
\centering
\includegraphics[scale=1]{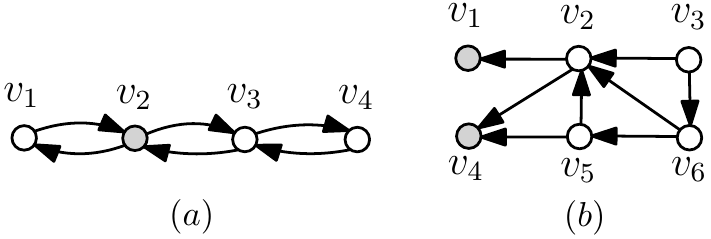}
\caption{Two networks and their leaders show in gray. For the network in (a),  $\delta(G,V_\ell)=3$, $\zeta(G,V_\ell)=1$.  For the network in (b),  $\delta(G,V_\ell)=5$, $\zeta(G,V_\ell)=6$. 
}
\label{fig:pvszex}
\end{figure}

\subsection {Advantages of Using the Distance-based Bound}
We will present two results, Theorems \ref{P>Z} and \ref{P>ZZ}, identifying some rich cases where $
\delta(G,V_\ell) >\zeta(G,V_\ell)$.  Later in Section \ref{sims}, we will also provide numerical results showing that $
\delta(G,V_\ell)$ is actually significantly greater than $\zeta(G,V_\ell)$ in many cases that are not limited to those captured by Theorems \ref{P>Z} and \ref{P>ZZ}. Our first result in this section shows that $
\delta(G,V_\ell)$ is greater than $\zeta(G,V_\ell)$ whenever each leader has at least two followers as in-neighbors. Note that this condition is very likely to occur when a small number of leaders are scattered over a large graph where most nodes have an in-degree of two or more (e.g., most regular graphs, random graphs, scale-free networks). 
\begin{theorem}
\label{P>Z}
Consider any graph $G = (V,E)$ with $n$ nodes and $m$ leaders $V_\ell \subseteq V$. If each leader has at least two followers as in-neighbors, then $
\delta(G,V_\ell) > \zeta(G,V_\ell)$.
\end{theorem}
\begin{proof}
If every leader has incoming links from at least two followers, then none of the followers will be forced when only the leaders are the black nodes. Accordingly, the $dset(G,V_\ell)=V_\ell$ and $
\zeta(G,V_\ell)= m$.
On the other hand, we can always find a PMI sequence of DL vectors whose length is greater than $m$ in such a case. As an example, consider the following sequence that has a length of $m+1$: 1) start with the DL vectors of leaders in any order, 2) add the DL vector of a follower who has a distance of one to one of the leaders. Since each leader is the only node who has a distance of zero to itself, those self-distance entries can be selected as the entries that satisfy the PMI rule. Hence, the longest possible PMI sequence would have a length of at least $m+1$, which implies $\delta(G,V_\ell) > \zeta(G,V_\ell)$.
\end{proof}

Our next result shows that for any single-leader network where each follower has a finite distance to the leader, $\delta(G,V_\ell) < n$ ensures that $\delta(G,V_\ell) > \zeta(G,V_\ell)$. 
\begin{theorem}
\label{P>ZZ}
For any $G = (V,E)$ with $n$ nodes and a single leader $v_l \in V$ such that $d(v_i,v_l)<\infty$ for all $v_i  \in V$,
\begin{equation}
\label{P>ZZ1}
\delta(G,V_\ell) < n \Rightarrow \delta(G,V_\ell)> \zeta(G,V_\ell).
\end{equation}
\end{theorem}
\begin{proof}
Since the left side of \eqref{P>ZZ1} can never be true for $n=1$, we focus on networks with $n\geq2$ and we will prove the claim via contradiction.  Suppose that $\delta(G,V_\ell) < n$ and $\zeta(G,V_\ell)\geq \delta(G,V_\ell)$. Note that if $v_l$ has more than one follower as in-neighbor, then the zero forcing process starting with the input set $\{v_l\}$ would not propagate and we would have $\zeta(G,V_\ell)=1$. Furthermore, for any network with a single leader $v_l \in V$ such that $d(v_i,v_l)<\infty$ for all $v_i  \in V$,
\begin{equation}
\label{P>ZZ2}
\delta(G,\{v_l\}) = \max_{v_i \in V}d(v_l,v_i)+1,
\end{equation}
which is always greater than one. Hence, if $\zeta(G,V_\ell)\geq \delta(G,V_\ell)$, then $v_l$ must have only one in-neighbor, say $v_i$, who will be infected by $v_l$ under the zero forcing process. Now, if $n=2$ (there are no other followers), then we end up with $\delta(G,V_\ell) = \zeta(G,V_\ell) =2$, which contradicts with $\delta(G,V_\ell) < n$. On the other hand, if $n>2$ then we can repeat the same reasoning by removing $v_l$ from the network, since $v_l$ has no impact on the infection of nodes at distance of two or more from itself, and treating the remaining network as a system with a single leader $v_i$ with $d(v_j,v_i)<\infty$ for every $v_j\neq v_l$ ($v_i$ being the only in-neighbor of $v_l$ implies that the paths from all other nodes to $v_l$ goes through $v_i$, hence $d(v_j,v_i)<\infty$).  Accordingly, we can show that if $\zeta(G,V_\ell)\geq \delta(G,V_\ell)$, then each follower must have a distinct distance from $v_l$, which implies $\delta(G,V_\ell) = \zeta(G,V_\ell) =n$ and results in a contradiction with $\delta(G,V_\ell) < n$.
\end{proof}

\begin{remark}
\label{remsingle}
In light of \eqref{P>ZZ2}, the only connected undirected network with a single-leader that yields $\delta(G,V_\ell) = n$ is a path graph with a terminal node being the leader. Hence, Theorem \ref{P>ZZ} implies that for all other connected undirected networks with a single-leader, we have $\delta(G,V_\ell) > \zeta(G,V_\ell)$.
\end{remark}


\subsection{Advantages of Using the Zero-forcing-based Bound}
Here, we show that one major advantage of using the zero-forcing-based approach is that it is better at verifying complete strong structural controllability. More specifically, we show that if $\delta(G,V_\ell)  = n$, then  $V_\ell$ must be a zero forcing set. Note that the converse is not true in general, i.e., it is possible to have a zero forcing set $V_\ell$ such that $\delta(G,V_\ell)  < n$, as already shown by the example in Fig. \ref{fig:pvszex}b. Clearly, such examples do not exist for single-leader networks due to Theorem \ref{P>ZZ}.

\begin{theorem}
\label{fullPMI-ZFS}
For any graph $G = (V,E)$ with $n$ nodes and any set of $m$ leaders $V_\ell \subseteq V$,
\begin{equation}
\label{FPZ1}
\delta(G,V_\ell)  = n \Rightarrow  \zeta(G,V_\ell) = n.
\end{equation}
\end{theorem}
\begin{proof} The claim is trivial for the cases when $V_\ell = V$ since $\delta(G,V)  = \zeta(G,V) = n$. Hence we focus on $V_\ell \subset V$ ($n>m$) in the proof. Let $\calD=[\calD_1 \; \calD_2 \; \cdots \; \calD_n]$ be a PMI sequence consisting of all the distance-to-leaders (DL) vectors such the first  $|V_\ell|$ vectors belong to the leaders. Note that there is no loss of generality here since for any PMI sequence of DL vectors, the vectors belonging to the leaders can be moved to the beginning of the sequence and the distance of each leader to itself (zero) satisfies the PMI rule. Without any loss of generality, let the nodes be re-labeled based on the order of their DL vectors in the sequence, i.e., $\calD_i$ is the DL vector of $v_i \in V$ for all $i=1, 2, \hdots, n$. Furthermore, let $\pi(i)$ denote the dimension of $\calD_i$ that satisfies the PMI rule, i.e., 
\begin{equation}
\label{FPZ2}
[\calD_i]_{\pi(i)} < [\calD_{j}]_{\pi(i)}, \;\;\forall j > i.
\end{equation}
Due to Lemma 4.1 in \cite{yazicioglu2016graph}, if $\calD$ is the longest possible PMI sequence of DL vectors, then it must satisfy
\begin{equation}
\label{FPZ3}
[\calD_i]_{\pi(i)} = \min_{j \geq i} [\calD_{j}]_{\pi(i)}, \forall i \in \{ 1, \hdots, n-1\}.
\end{equation}
For each $i \in \{m+1, \hdots, n\}$, let $W_i =\{v_i, \hdots, v_n\} \subseteq V$ be the owners of the DL vectors in the subsequence of $\calD$ starting with the $i^{th}$ entry. We will show that 
\begin{equation}
\label{FPZ4}
\forall i> m, \exists k<i : \mathcal{N}_k \cap W_i = \{v_i\},
\end{equation}
where $\mathcal{N}_k$ is the set of  in-neighbors of $v_k$. Note that \eqref{FPZ4} would imply that if all the nodes $\{v_1, \hdots, v_{i-1}\}$ are infected, then $v_i$ becomes infected under the zero-forcing process. Accordingly, we can conclude that $\zeta(G,V_\ell) = n$ since starting with all the leaders being infected, all the followers would eventually become infected.

Note that \eqref{FPZ4} clearly holds for $i=n$ since $W_n = \{v_n\}$ and $v_n$ must have at least one out-neighbor in $\{v_1, \hdots, v_{n-1}\}$ as otherwise its DL vector would be all $\infty$ and not included in any PMI sequence, leading to the contradiction $\delta(G,V_\ell)<n$. Now, for the sake of contradiction, suppose that \eqref{FPZ4} is not true for some $i \in \{m+1, \hdots, n-1\}$. Let $v_k$ be any out-neighbor of $v_i$ such that
\begin{equation}
\label{FPZ5}
[\calD_k]_{\pi(i)}=[\calD_i]_{\pi(i)}-1.
\end{equation}
Clearly such a neighbor always exists: $v_k$ is either the leader $l_{\pi(i)}$ or another follower on the shortest path from $v_i$ to $l_{\pi(i)}$. Furthermore, $k<i$ due to \eqref{FPZ2}. Now suppose that $v_k$ has another in-neighbor $v_j$ such that $j>i$. Then, 
\begin{equation}
\label{FPZ6}
[\calD_j]_{\pi(i)}\leq[\calD_k]_{\pi(i)}+1 = [\calD_i]_{\pi(i)},
\end{equation}
which contradicts with \eqref{FPZ2}. Hence, \eqref{FPZ4} must be true, and it implies that $\zeta(G,V_\ell) = n$.
\end{proof}


\section{Combined Bound: $\delta(G, \text{dset}(G,V_\ell))$}
\label{combine}
Our analysis so far has shown that both the distance-based bound, $\delta(G,V_\ell)$, and the zero-forcing-based bound, $\zeta(G,V_\ell)$, have their own merits.  Given these results, it is only natural to ask if it is possible to find a novel bound that combines the strengths of distance-based and zero-forcing-based methods. In this regard, one trivial approach is taking the maximum of the two bounds.
While guaranteed to be at least as good as either of the bounds alone, this approach does not reveal any additional information compared to the two original bounds. In this section, we present a novel bound that fuses the strengths of distance-based and zero-forcing-based approaches. More specifically, we show that the length of the longest PMI sequence of distances to the derived set of leaders, i.e.,
\begin{equation}
\label{mbound2}
\delta(G, \text{dset}(G,V_\ell)),
\end{equation}
provide a tight lower bound on the dimension of SSCS. We show that this novel bound is always at least as good as, and sometimes greater than, either of the bounds alone. To this end, we first provide a result on the invariance of SSCS in diffusively coupled networks to the addition of every node in the derived set,  $\text{dset}(G,V_\ell)$, as leaders.

\begin{theorem}\cite{monshizadeh2015strong}
For any network $G=(V,E)$ with the leaders $V_\ell \subseteq V$, and any weight function $w:E\rightarrow \mathbb{R}^+$,
\begin{equation}
\label{dteq}
range(\Gamma(L_w,V_\ell))= range(\Gamma(L_w,\text{dset}(G,V_\ell))), 
\end{equation}
where $range(\Gamma)$ is the range space of controllability matrix.
\end{theorem}
\begin{proof}
The proof follows from Lemma 4.1 in \cite{monshizadeh2015strong}, which shows a stronger condition, i.e., \eqref{dteq} holds for a set of state matrices that contain weighted Laplacians as a subset.  
\end{proof}

\begin{theorem}
\label{cbound}
Consider any network $G=(V,E)$ with the leaders $V_\ell \subseteq V$. Then,
\begin{equation}
\label{cbound1}
\delta(G,V_\ell), \zeta(G,V_\ell) \le \delta(G,\text{dset}(G,V_\ell)) \le \gamma(G,V_\ell).
\end{equation}
\end{theorem}
\begin{proof}
First, we show that $\delta(G,\text{dset}(G,V_\ell)) \le \gamma(G,V_\ell)$.
In light of \eqref{eq:dim_SSC} and \eqref{dteq},
\begin{equation}
\label{cbound2}
\gamma(G,\text{dset}(G,V_\ell)) = \gamma(G,V_\ell).
\end{equation}
Due to Theorem \ref{thm:PMI},
\begin{equation}
\label{cbound3}
\delta(G,\text{dset}(G,V_\ell)) \le \gamma(G,\text{dset}(G,V_\ell)).
\end{equation}
Using \eqref{cbound2} and \eqref{cbound3}, we get ${\delta(G,\text{dset}(G,V_\ell)) \le \gamma(G,V_\ell)}$.

Next, we show that $\delta(G,\text{dset}(G,V_\ell)) \ge \zeta(G,V_\ell)$. Since the DL vectors of leaders can always be included in the beginning of a PMI sequence (self-distances are uniquely zero), $\delta(G,V') \geq |V'|$ for any $V' \subseteq V$. Hence, 
\begin{equation}
\label{cbound4}
\delta(G,\text{dset}(G,V_\ell)) \ge |\text{dset}(G,V_\ell)|= \zeta(G,V_\ell).
\end{equation}

Finally, we show that $\delta(G,\text{dset}(G,V_\ell)) \ge \delta(G,V_\ell)$. Since the initial set of infected nodes (input nodes) are always contained in the derived set, we have $V_\ell \subseteq \text{dset}(G,V_\ell)$. 
Accordingly, for any PMI sequence $\mathcal{D}$ of DL vectors under the leader set $V_\ell$, there is an equally long PMI sequence of DL vectors $\mathcal{D}'$ under the leader set $\text{dset}(G,V_\ell)$, which has the DL vectors of the same nodes in the same order as $\mathcal{D}$.  Hence, the longest possible PMI sequence of DL vectors with the additional leaders can not be shorter, i.e.,
\begin{equation}
\label{cbound5}
\delta(G,\text{dset}(G,V_\ell)) \ge \delta(G,V_\ell).
\end{equation}
\end{proof}

\begin{remark}
While Theorem \ref{cbound} shows that the combined bound is at least as good as the distance-based and zero-forcing-based bounds, it should also be emphasized that there exist networks $G=(V,E)$ and leader sets $V_\ell \subseteq V$, where the combined bound is strictly better than the two original bounds, i.e.,  
 $\delta(G,\text{dset}(G,V_\ell))>\delta(G,V_\ell), \zeta(G,V_\ell)$. We provide two such examples in Fig. \ref{fig:cbex}.
\begin{figure}[htb]
\centering
\includegraphics[scale=0.9]{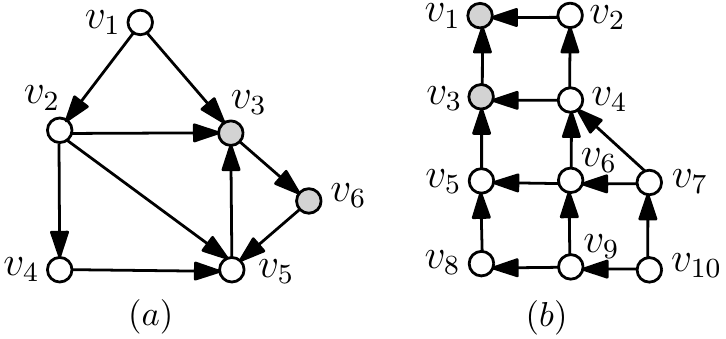}
\caption{Two networks and their leaders (gray). In (a): $\delta(G,\text{dset}(G,V_\ell))=5$, $\delta(G,V_\ell)=4$, $\zeta(G,V_\ell)=3$.  In (b): $\delta(G,\text{dset}(G,V_\ell))=9$, $\delta(G,V_\ell)=6$, $\zeta(G,V_\ell)=5$. 
}
\label{fig:cbex}
\end{figure}
\end{remark}

\section{Numerical Results}
\label{sims}
We compare the lower bounds on the dimension of strong structurally controllable subspace on Erd\"os-R\'enyi (ER) and Barab\'asi-Albert (BA) graphs. ER graphs are the ones in which any two nodes are adjacent with a probability $p$. BA graphs are obtained by adding nodes to an existing graph one at a time. Each new node is adjacent to $\varepsilon$ existing nodes that are chosen with probabilities proportional to their degrees.

In all the simulations, we consider undirected graphs with $n=100$ nodes. In Figs. \ref{fig:ER} and \ref{fig:BA}, we plot lower bounds on the dimension of SSCS, including $\delta(G,V_\ell)$, $\zeta(G,V_\ell)$ and $\delta(G,\text{dset}(V(_\ell))$, as a function of number of leaders $|V_\ell| = \ell$. We select the leader nodes randomly. Each point on the plots corresponds to the average of 100 randomly generated instances. While we computed the exact value of  $\zeta(G,V_\ell)$, we used the greedy approximation (underestimation) in \cite{shabbir2019computation} for computing $\delta(G,V_\ell)$ and $\delta(G,\text{dset}(V(_\ell))$ due to the large number of leaders. While this approximation was shown to be very close in \cite{shabbir2019computation}, the true gap between these two bounds and $\zeta(G,V_\ell)$ may be larger than shown in the plots. 

 In all the plots in Figs. \ref{fig:ER} and \ref{fig:BA}, we observe that the distance-based bound $\delta(G,V_\ell)$ starts above the ZFS-based bound $\zeta(G,V_\ell)$, which is expected due to Theorem \ref{P>ZZ} (or Remark \ref{remsingle}). Furthermore, $\delta(G,V_\ell)$ is usually significantly larger than $\zeta(G,V_\ell)$, especially when the number of leaders is small. This can be explained by Theorem \ref{P>Z} since most of the nodes in these networks have degrees of two or more. In the ER graphs the expected degree of each node is approximately $pn$, and each node in the BA graphs has a degree of $\varepsilon$ or more. Indeed, all the plots show a linear trend in $\zeta(G,V_\ell)$ when the number of leaders is small, indicating $\zeta(G,V_\ell)\approx|V_\ell|$. Note that when $\zeta(G,V_\ell)=|V_\ell|$, trivially $\delta(V,\text{dset}(V_\ell))=\delta(G,V_\ell)$, which explains why the distance-based and combined bounds mostly overlap until the number of leaders is sufficiently large and the zero-forcing-based bound departs from the initial linear regime. While the difference between the combined bound $\delta(V,\text{dset}(V_\ell))$ and $\delta(G,V_\ell)$ was observed to be  insignificant in these simulations, it is worth emphasizing that $\delta(V,\text{dset}(V_\ell))$ is the only bound guaranteed to be at least as good as the other two in all possible cases (Theorem \ref{cbound}) and the improvement with respect to $\delta(G,V_\ell)$ may be more significant for other families of networks. 
Finally, we see in all the plots that the three bounds approach each other as they all increase toward $n$, which is expected due to Theorem \ref{fullPMI-ZFS}.

\begin{figure}[!h]
\centering
\begin{subfigure}{0.23\textwidth}
\centering
\includegraphics[scale=0.295]{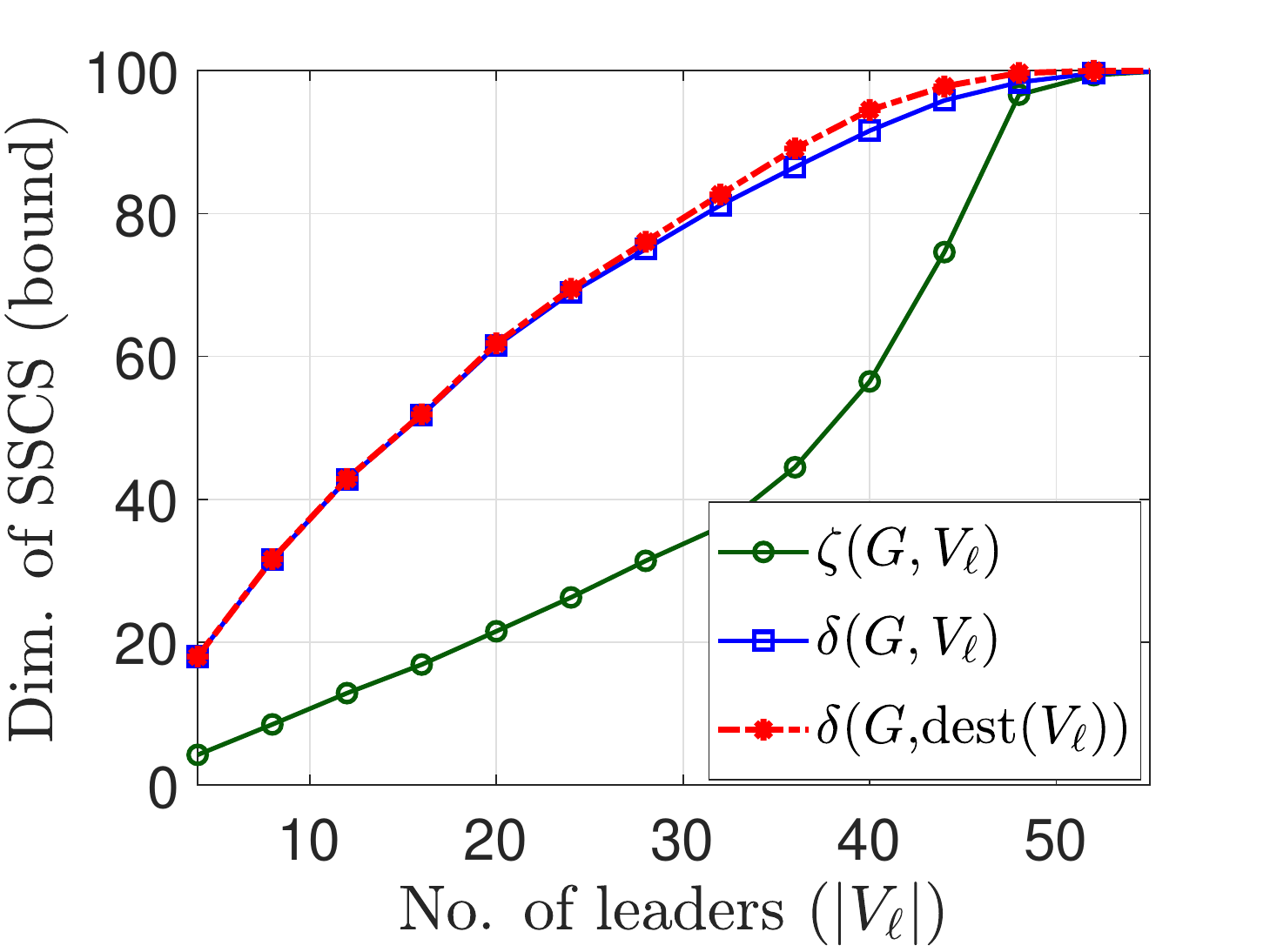}
\caption{$p=0.05$}
\end{subfigure}
\begin{subfigure}{0.22\textwidth}
\centering
\includegraphics[scale=0.29]{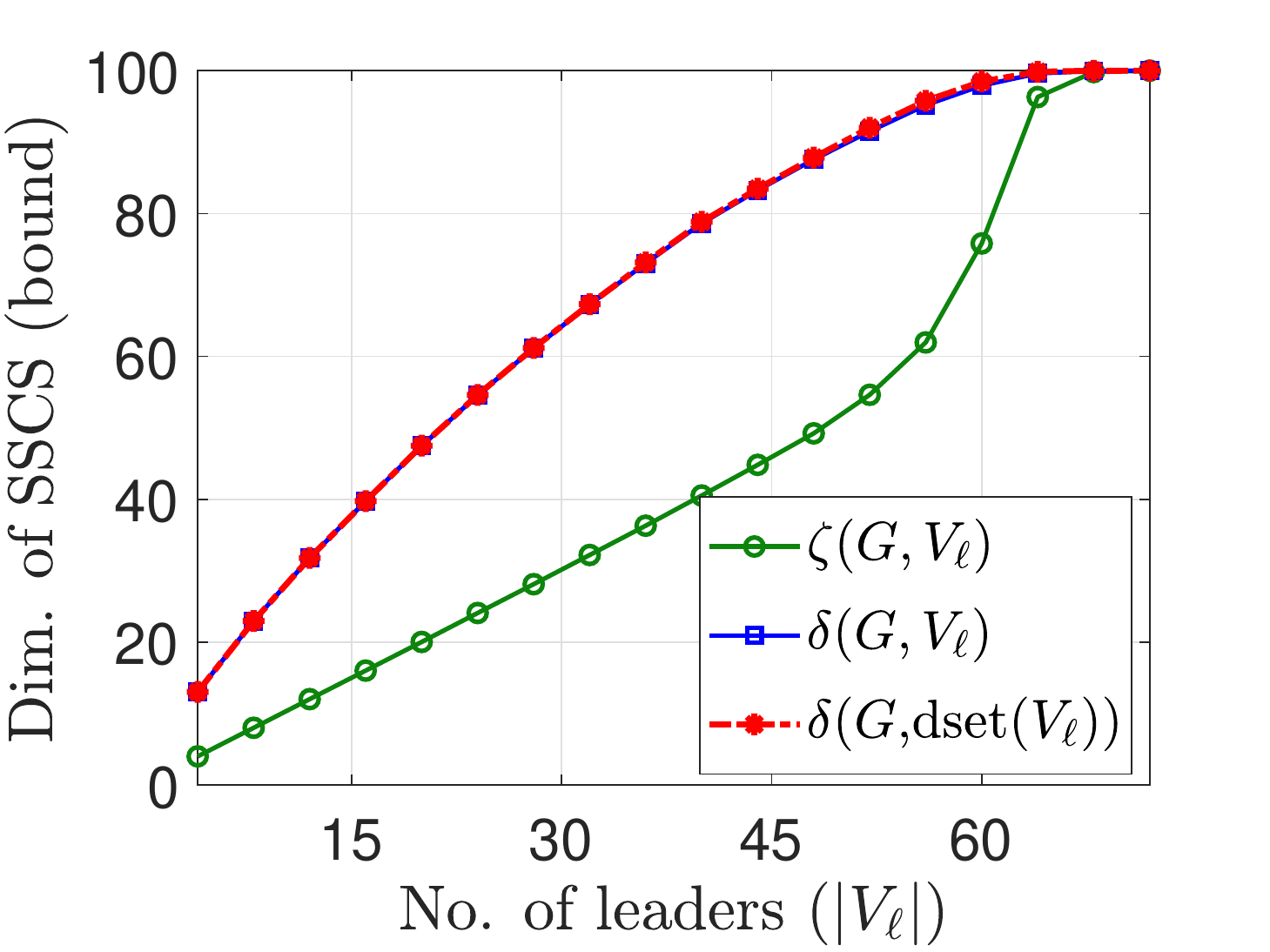}
\caption{$p=0.1$}
\end{subfigure}
\caption{Comparison of ZFS-based $\zeta(G,V_\ell)$, distance-based $\delta(G,V_\ell)$ and combined $\delta(G,\text{dset}(V_\ell))$ bounds on the dimension of SSCS in ER graphs.}
\label{fig:ER}
\end{figure}
\begin{figure}[htb]
\centering
\begin{subfigure}[b]{0.23\textwidth}
\centering
\includegraphics[scale=0.295]{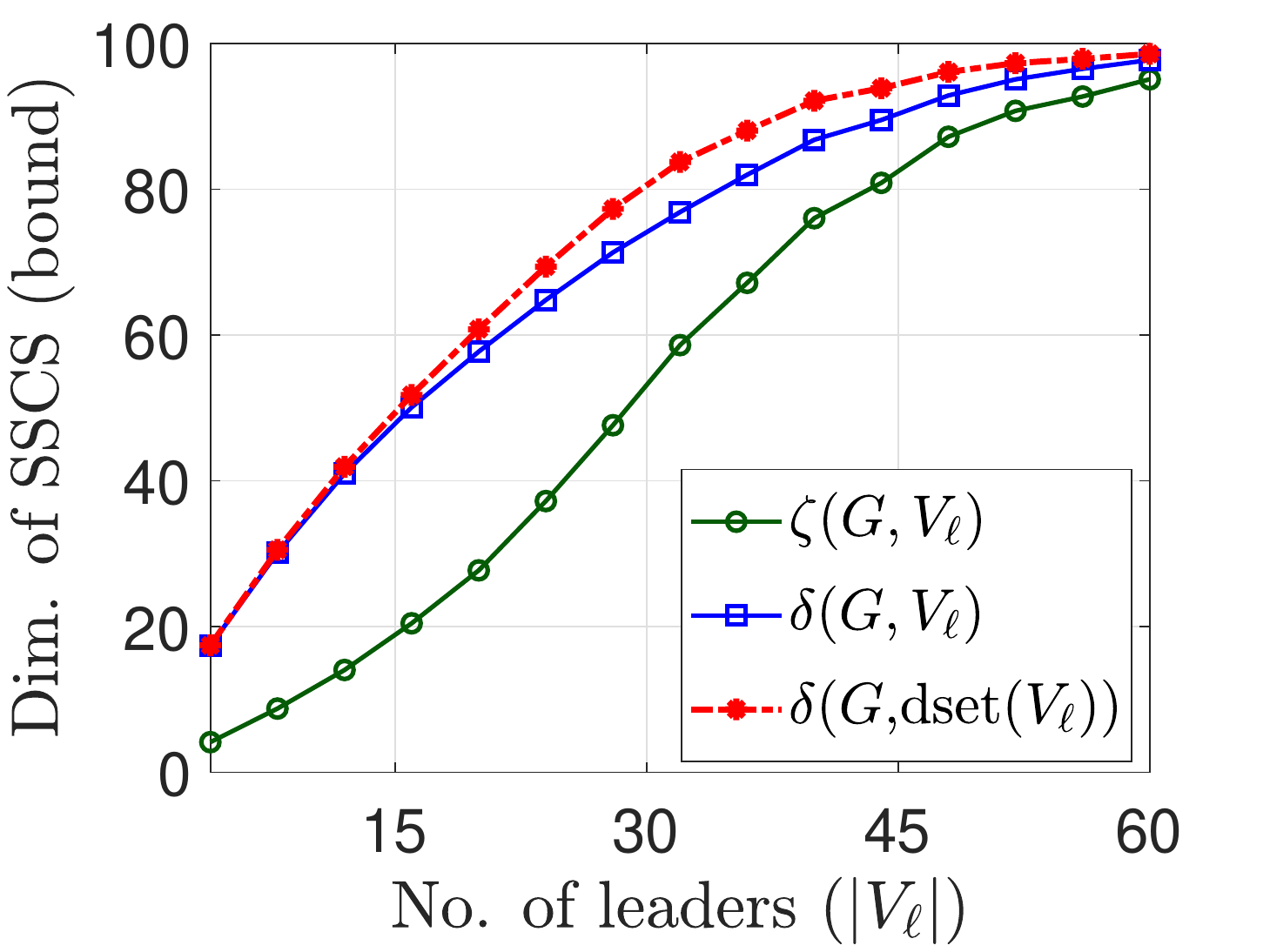}
\caption{$\varepsilon = 2$}
\end{subfigure}
\begin{subfigure}[b]{0.22\textwidth}
\centering
\includegraphics[scale=0.295]{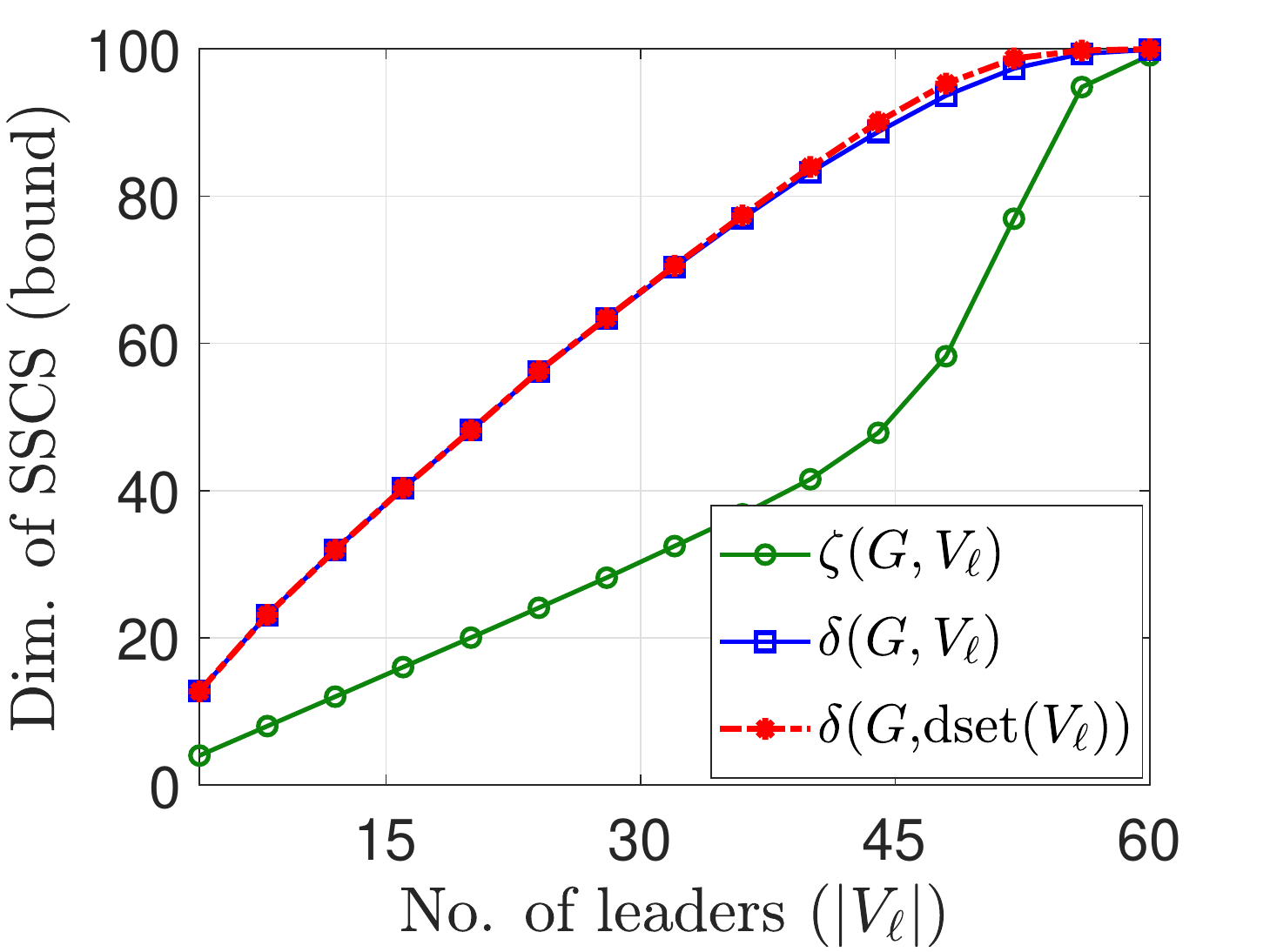}
\caption{$\varepsilon = 5$}
\end{subfigure}
\caption{Comparison of ZFS-based $\zeta(G,V_\ell)$, distance-based $\delta(G,V_\ell)$ and combined $\delta(G,\text{dset}(V_\ell))$ bounds on the dimension of SSCS in BA graphs.}
\label{fig:BA}
\end{figure}
\section{Conclusion}
\label{conc}

In this paper, we focused on the  the dimension of \emph{strong structurally controllable subspace (SSCS)} of networks under weighted Laplacian dynamics. We compared two tight lower bounds on the dimension of SSCS: one based on distances and the other based on zero forcing. We characterized various cases where the distance-based lower bound is guaranteed to be greater than the zero-forcing-based bound. On the other hand, we also show that, for any network of $n$ nodes, any set of leaders that makes the distance-based bound equal to $n$ is necessarily a zero forcing set. These results indicate that while the zero-forcing-based approach may be a better choice for verifying complete strong structural controllability, the distance-based approach is usually more informative when the leaders do not constitute a zero forcing set. We also present a novel bound based on the combination of these two approaches, which is always at least as good as, and in some cases strictly better than, the maximum of the two bounds. Finally, we numerically compared the bounds on various networks.

As a future direction, we plan to improve the proposed combined bound, for example by utilizing the invariance of controllable subspace to the addition/removal of links between leaders \cite{yazicioglu2013leader}. Obtaining a formal characterization of cases where the zero-forcing bound is guaranteed to be greater than the distance-based bound is another direction we plan to explore. Furthermore, the distance-based bound was recently utilized for analyzing the robustness-controllability trade-off in networks \cite{Abbas2019}. We intend to use the combined bound for further exploration of such trade-offs. 

\bibliographystyle{IEEEtran}
\bibliography{refer}
\end{document}